\documentclass{article}
\usepackage[letterpaper,margin=1in]{geometry}
\usepackage{graphicx} 
\usepackage{amsmath, amsthm, amssymb, amsfonts, braket, complexity, enumerate, mathrsfs, tikz, caption, subcaption, xcolor, mathtools, braket, dirtytalk}
\usepackage{bbm, dsfont}
\usepackage{thmtools} 
\usepackage{thm-restate}

\usepackage{tabularx}

\usepackage[linesnumbered,ruled,vlined]{algorithm2e}
\SetKwInput{KwInput}{Input}                
\SetKwInput{KwOutput}{Output}
\SetKwFor{RepTimes}{repeat}{times}{end}

\usepackage[pagebackref,colorlinks,citecolor=blue,linkcolor=magenta,bookmarks=true]{hyperref}
\usepackage[nameinlink,capitalize]{cleveref}

\newcommand{\wh}[1]{\widehat{#1}}

\newcommand{\F}{\mathbb{F}}
\renewcommand{\hat}{\widehat}

\theoremstyle{plain}
\newtheorem{theorem}{Theorem}[section]
\newtheorem{corollary}[theorem]{Corollary}
\newtheorem{definition}[theorem]{Definition}
\newtheorem{proposition}[theorem]{Proposition}
\newtheorem{remark}[theorem]{Remark}
\newtheorem{lemma}[theorem]{Lemma}
\newtheorem{claim}[theorem]{Claim}
\newtheorem{fact}[theorem]{Fact}

\newtheorem{problem}[theorem]{Problem}

\newcommand{\weyl}{\mathrm{Weyl}}

\newcommand{\eps}{\varepsilon}

\renewcommand{\Pr}{\mathop{\bf Pr\/}}
\newcommand{\Ex}{\mathop{\bf E\/}}

\newcommand{\tr}{\mathrm{tr}}  \newcommand{\trace}{\tr}

\newcommand{\calF}{\mathcal{F}}

\newcommand{\calP}{\mathcal{P}}

\newcommand{\abs}[1]{\lvert #1 \rvert}

\newcommand{\ketbra}[2]{\ket{#1}\!\!\bra{#2}}

\newcommand{\sympcomp}{\perp}
\newcommand{\indic}[1]{\mathbbm{1}_{#1}}   

\title{Tolerant Testing of Stabilizer States with Mixed State Inputs}
\author{Vishnu Iyer\thanks{University of Texas at Austin. \texttt{vishnu.iyer@utexas.edu}.} \and Daniel Liang\thanks{Portland State University. \texttt{danliang@pdx.edu}}}
\date{\today}

\begin{document}

	\maketitle
	
	\begin{abstract}
		We study the problem of tolerant testing of stabilizer states. In particular, we give the first such algorithm that accepts mixed state inputs. Formally, given a mixed state $\rho$ that either has fidelity at least $\eps_1$ with some stabilizer pure state or fidelity at most $\eps_2$ with all such states, where $\eps_2 \leq \eps_1^{O(1)}$, our algorithm distinguishes the two cases with sample complexity $\poly(1/\eps_1)$ and time complexity $O(n \cdot \poly(1/\eps_1))$.
	\end{abstract}
	
	\section{Introduction}
	In the property testing model, one is given an input state $\rho$ and a class of states $\calP$, known as a \emph{property}, and must decide whether $\rho$ is in $\calP$ or far from every state in $\calP$. We refer to \cite{montanaro2016survey} for an introduction to quantum property testing. This task can be generalized to \emph{tolerant} testing (which was introduced classically in \cite{parnas2006tolerant}): is $\rho$ $\eps_1$-close to something in $\calP$ or $\eps_2$-far from everything in $\calP$? Formally, denoting the fidelity between states $\rho$ and $\sigma$ as $\calF(\rho, \sigma)$, the problem is as follows:
	\begin{problem}
		Fix a set of quantum states $\calP$ (this is usually referred to as a \emph{property}) that is known in advance. Given copies of an unknown quantum state $\rho$ and parameters $\eps_1$ and $\eps_2$ such that $1 \geq \eps_1 > \eps_2 \geq 0$, decide if $\sup_{\sigma \in \calP} \calF(\rho, \sigma) \geq \eps_1$ or $\sup_{\sigma \in \calP} \calF(\rho, \sigma) \leq \eps_2$ promised that one of the cases is true. 
	\end{problem}
	
	Significant process has been made when $\calP$ is the set of stabilizer states.
	The fidelity to the closest stabilizer state $\ket \phi$ is commonly referred to as the \emph{stabilizer fidelity}
	\[
	\calF(\rho) \coloneqq \max_{\ket \phi} \braket{\phi | \rho  | \phi}.
	\]
	The first test was introduced by \cite{gross2021schur} (henceforth, the \emph{GNW test}) for when $\eps_1 = 0$, with follow-up work by \cite{damanik2018optimality}.
	\cite{grewal_et_al:LIPIcs.ITCS.2023.64,grewal2023improved} made the first steps towards tolerant testing by improving the completeness beyond $\eps_1 = 0$ for GNW test, but the soundness analysis has limitations in certain regimes.
	This has now been (partially) remedied by \cite{arunachalam2024toleranttestingstabilizerstates,arunachalam2024notepolynomialtimetoleranttesting,bao2024toleranttestingstabilizerstates,mehraban2024improvedboundstestinglow} by giving an improved soundness analysis in the `far' setting of stabilizer state property testing.
	
	We note that \cite{haug2023scalable} gave a related test with perfect completeness, but (to the authors' knowledge) no rigorous proof of soundness exists.
	\cite{bu2023stabilizertestingmagicentropy} derive a quantum generalization of the convolution and give a tolerant test with identical performance to the GNW test on pure states.
	The test works by making a pure state become more mixed the farther it is from a stabilizer state.
	The purity of the resulting state can then be measured using a SWAP test.
	\cite{hinsche2024singlecopystabilizertesting} interestingly gave the first test that only uses single-copy measurements with provable soundness guarantees, but (to the authors' knowledge) no provable completeness analysis beyond $\eps_1 = 0$ has been given, and it also has an inherent system-size dependence, unlike the work using the GNW test.
	
	\subsection{This Work}
	
	The major caveat is that in the analysis of \emph{all} of the previous works on stabilizer state property testing, the input state was assumed to be a pure state.
	Due to the inherent nature of property testing and noise, an algorithm that works when given mixed state inputs is desirable.
	To that end, we propose a new test, based on the GNW test, that provably works for mixed state inputs.
	
	\begin{theorem}
		Let $\eps_1, \eps_2 \in [0, 1]$ such that $\eps_2 \leq \eps_1^{O(1)}$. There exists an algorithm that takes $\poly(1/\eps_1)$ copies of $\rho$ and $n \cdot \poly(1/\eps_1)$ time and can perform property testing in this scenario with success probability greater than $2/3$ even when the input is a mixed state.
	\end{theorem}
	
	Our new test and accompanying proof has the following properties:
	\begin{enumerate}
		\item Performs identically to that of \cite{gross2021schur} and follow-up work \cite{grewal_et_al:LIPIcs.ITCS.2023.64,grewal2023improved,arunachalam2024toleranttestingstabilizerstates,arunachalam2024notepolynomialtimetoleranttesting,bao2024toleranttestingstabilizerstates,mehraban2024improvedboundstestinglow} on pure states,
		\item Provides equivalent scaling in terms of stabilizer fidelity compared to that of recent art \cite{grewal2023improved,arunachalam2024notepolynomialtimetoleranttesting,bao2024toleranttestingstabilizerstates}, even when given mixed states.
		\item Maintains the same asymptotic runtime per iteration and hardware requirements as the previous tests.
	\end{enumerate}
	We establish this result by generalizing and modifying the proofs of tolerant stabilizer testing in the pure state input setting of \cite{grewal2023improved,arunachalam2024toleranttestingstabilizerstates,arunachalam2024notepolynomialtimetoleranttesting, bao2024toleranttestingstabilizerstates} to when the input state is mixed.
	We also show that the symplectic Fourier spectrum of the Pauli decomposition for mixed states has interesting properties that likely warrant further study.
	
	We note that the GNW test has additional nice properties such as 1) uses $6$ copies per iteration, 2) runs in linear time per iteration, 3) uses only $2$ copies at any given time, 4) only uses Clifford gates, and 5) no additional ancilla qubits.
	Our new test matches the GNW test in these properties.
	Additionally, for any pure input state $\ket \psi$, the two tests have identical acceptance probabilities, so that any future updates to pure state input analyses will immediately extend to this test.
	
	\begin{remark}
		The bounds that we achieve are equivalent to \cite{grewal_et_al:LIPIcs.ITCS.2023.64,grewal2023improved} in the completeness case, and those of \cite{arunachalam2024toleranttestingstabilizerstates,arunachalam2024notepolynomialtimetoleranttesting,bao2024toleranttestingstabilizerstates} in the soundness case.\footnote{They also match \cite{gross2021schur,grewal2023improved} in the close regime. See \cref{appendix:close}.}
		That said, there has been no real attempt made by any work yet (including this one) to optimize these constants.
		
		During the writing this note,  \cite{mehraban2024improvedboundstestinglow} has since come out with a better dependence on the exponent for soundness.
		It is not immediately clear (to the authors) how to generalize their proofs to the mixed state case, since it (at the intuitive level) relies on a connection with phase states, which are inherently pure states.
	\end{remark}

	\paragraph{Related Works}
	After writing this note, we became aware that the measurement that our test implements is identical to that of \cite[Algorithm 1]{Haug2024efficient} with $n=3$.
	However, their analysis is solely for pure states, as they use the property that $p_\Psi$, the squares of the coefficients in the Pauli decomposition, forms a distribution to then study the entropy of $p_\Psi$ (i.e., the stabilizer entropy).
	Much of their analysis is very specific to properties of $p_\Psi$ that do not hold when moving to $p_\rho$.

	\section{Preliminaries}
	\subsection{Symplectic Fourier Analysis}
	
	\begin{definition}[Symplectic Product]
		For $x = (a, b) \in \F_2^{2n}$ and $y = (c, d) \in \F_2^{2n}$, the symplectic product \[[x, y] \coloneqq a\cdot d + b \cdot c. \]
	\end{definition}
	
	\begin{definition}[Symplectic Complement]
		Let $A \subseteq \F_2^{2n}$ be a subspace. The \emph{symplectic complement} of $A$, denoted by $A^\sympcomp$, is defined as
		\[
		A^\sympcomp \coloneqq \{x \in \F_2^{2n} : \forall a \in A, [x, a] = 0\}.
		\]
	\end{definition}
	
	An isotropic subspace is any subspace $V \subset \F_2^{2n}$ where $V \subseteq V^\sympcomp$.
	A subspace is Lagrangian if $V = V^\sympcomp$.
	Finally, a subspace is symplectic if $V \cap V^{\sympcomp} = \{0^{2n}\}$.
	
	\begin{definition}[Symplectic Fourier transform]
		Let $f : \F_2^{2n} \rightarrow \mathbb{R}$. We define the symplectic transform of $f$ by $\hat{f} : \F_2^{2n} \rightarrow \mathbb{R}$ such that
		\[
		\hat{f}(a) \coloneqq \frac{1}{4^n} \sum_{x \in \F_2^{2n}} (-1)^{[a, x]} f(x).
		\]
	\end{definition}
	
	It follows that $f(x) = \sum_{a \in \F_2^{2n}} (-1)^{[a, x]}\hat{f}(a)$.
	
	We require the following facts that resemble their standard Boolean Fourier analysis counterparts.
	See \cite[Section 2.2]{grewal2023improved} for omitted proofs.

	\begin{fact}[Plancherel's Theorem]\label{fact:plancherel}
		\[
		\frac{1}{4^n} \sum_{x \in \F_2^{2n}} f(x) g(x) = \sum_{x \in  \F_2^{2n}} \widehat{f}(x)\widehat{g}(x).
		\]
	\end{fact}
	
	\begin{definition}[Convolution]
		Let $f, g : \F_2^{2n} \to \mathbb{R}$. Their convolution is the function $f \ast g:\F_2^{2n} \to \mathbb{R}$ defined by 
		\[
		(f \ast g)(x) = \Ex_{t \sim \F_2^{2n}}[f(t) g(t + x)] = \frac{1}{4^n} \sum_{t \in \F_2^{2n}} f(t) g(t + x).
		\]
	\end{definition}
	
	\begin{fact}[{\cite[Proposition 2.10]{grewal2023improved}}]\label{thm:convolution-theorem}
		Let $f, g: \F_2^{2n} \to \mathbb{R}.$ Then for all $a \in \F_2^{2n}$, 
		\[
		\wh{f \ast g}(a) = \wh{f}(a) \wh{g}(a).
		\]
	\end{fact}
	
	\begin{lemma}[{\cite[Lemma 2.11]{grewal2023improved}}]\label{lemma:sum-over-characters}
		For any subspace $T \subseteq \F_2^{2n}$ and a fixed $x \in \F_2^{2n}$,
		\[
		\sum_{a \in T} (-1)^{[a, x]} = \abs{T} \cdot \indic{x \in T^\sympcomp}.
		\]
	\end{lemma}
	
	\begin{fact}\label{fact:duality}
		For any subspace $T \subseteq \F_2^{2n}$ and $f : \F_2^{2n} \rightarrow \mathbb{R}$,
		\[
		\sum_{x \in T} f(x) = \abs{T}\sum_{x \in T^\sympcomp} \hat{f}(x)
		\]
	\end{fact}
	\begin{proof}
		\begin{align*}
			\sum_{x \in T} f(x) &= \sum_{x \in T} \sum_{a \in \F_2^{2n}} (-1)^{[a, x]} \hat{f}(a)\\
			&=  \sum_{a \in \F_2^{2n}}  \hat{f}(a)\sum_{x \in T}(-1)^{[a, x]}\\
			&=  \abs{T} \sum_{a \in T^\sympcomp} \hat{f}(a). && (\mathrm{\cref{lemma:sum-over-characters}}) \qedhere
		\end{align*}
	\end{proof}
	
	\begin{fact}\label{fact:double-hat}
		For all $f : \F_2^{2n} \rightarrow \mathbb{R}$,
		\[\wh{\wh{f}} = \frac{f}{4^n}.\]
	\end{fact}
	\begin{proof}
		\begin{align*}
			\hat{\hat{f}}(x) &= \frac{1}{4^n}\sum_{a \in \F_2^{2n}} (-1)^{[a, x]} \hat{f}(a)\\
			&= \frac{1}{16^n}\sum_{a,b  \in \F_2^{2n}} (-1)^{[a, x + b]} f(b)\\
			&= \frac{1}{4^n} f(x) && (\mathrm{\cref{lemma:sum-over-characters}}) \qedhere
		\end{align*}
	\end{proof}
	
	\subsection{Weyl Operators}
	
	For $x = (a, b) \in \F_2^{2n}$ where $a$ and $b$ are the first and last $n$ bits of $x$, respectively, let $W_x$ refer to the Weyl operator
	\[
	W_x \coloneqq i^{a^\prime \cdot b^\prime} \bigotimes_{j=1}^n X^{a_j} Z^{b_j},
	\]
	where $a^\prime$ and $b^\prime$ are the embedding of $a$ and $b$ into $\mathbb{Z}^n$ respectively.
	The Weyl operators form an orthogonal basis for $\mathbb{C}^{2^n \times 2^n}$ so they can be used to give the Weyl decomposition of any Hermitian matrix
	\[
	H = \frac{1}{2^n} \sum_{x \in \F_2^{2n}} \trace\left(W_x H \right) W_x.
	\]
	
	The following is a well-known fact about commutation relationships of Pauli/Weyl operators that ties them to the symplectic product.
	
	\begin{fact}\label{fact:weyl-anticommute}
		\[W_a W_x W_a = (-1)^{[a, x]} W_x \]
	\end{fact}
	
	There is a version of Plancherel's theorem for the Weyl decomposition.
	\begin{fact}\label{fact:parseval-weyl}
		For matrices $A$ and $B$
		\[
		\trace(A^\dagger B) = \frac{1}{2^n}\sum_{x \in \F_2^{2n}} \trace(A^\dagger W_x) \cdot \trace(B W_x).
		\]
	\end{fact}
	
	We will note the normalized squares of the Weyl coefficients as \[p_H(x) \coloneqq \frac{1}{2^n}\trace^2\left(W_x H \right).\]
	
	From \cref{fact:parseval-weyl} one can see that for mixed state $\rho$, $\sum_{x \in \F_2^{2n}} p_\rho(x) = \tr(\rho^2)$
	and is thus only a distribution when $\rho$ is a pure state.
	Note also that $0 \leq p_\rho(x) \leq \frac{1}{2^n}$.

	The following is a generalization of \cite[Eqn 3.5]{gross2021schur} to general Hermitian matrices.
	
	\begin{lemma}\label{lem:p-hat}
		For Hermitian matrix $H$,
		\[
		\hat{p}_H(a) = \frac{1}{4^n} \tr(W_a \rho W_a \rho).
		\]
	\end{lemma}
	\begin{proof}
		\begin{align*}
			\hat{p}_H(a) \coloneqq& \frac{1}{4^n} \sum_{x \in \F_2^{2n}} (-1)^{[a, x]} p_H(x)\\
			=& \frac{1}{8^n} \sum_{x \in \F_2^{2n}} (-1)^{[a, x]} \tr(W_x H) \tr(W_x H) \\
			=&  \frac{1}{8^n} \sum_{x \in \F_2^{2n}} \tr(W_a W_x W_a H) \tr(W_x H) && (\mathrm{\cref{fact:weyl-anticommute}})\\
			=& \frac{1}{4^n} \tr(W_a H W_a H). && (\mathrm{\cref{fact:parseval-weyl}}) \qedhere
		\end{align*}
	\end{proof}
	
	\begin{remark}\label{remark:p-hat-pure}
		For rank $1$ matrices (i.e., pure states) it follows that $\hat{p}_H = \frac{p_H}{2^n}$, as famously pointed out by \cite{gross2021schur}.
		This is untrue for mixed states in general, leading to the general challenge of replacing the proofs in previous literature without using this very useful identity.
	\end{remark}
	
	The following two facts establish that $2^n \hat{p}_\rho$ is a distribution.
	
	\begin{fact}\label{fact:p-hat-non-negative}
		For mixed state $\rho$,
		\[
		0 \leq \hat{p}_\rho(x) \leq \frac{1}{4^n}
		\]
		for all $x \in \F_2^{2n}$.
	\end{fact}
	\begin{proof}
		Let $\rho = \sum_{i} \alpha_i \ketbra{\psi_i}{\psi_i}$, such that $\alpha_i \geq 0$ and $\sum_i \alpha_i = 1$.
		\begin{align*}
			4^n \hat{p}_\rho(x) \coloneqq& \trace\left(W_x \rho W_x \rho\right)\\
			=& \trace\left(W_x \sum_{i} \alpha_i \ketbra{\psi_i}{\psi_i} W_x \sum_{i} \alpha_j \ketbra{\psi_j}{\psi_j}\right)\\
			=& \sum_{i, j} \alpha_i \alpha_j \trace\left(W_x \ketbra{\psi_i}{\psi_i} W_x \ketbra{\psi_j}{\psi_j} \right) && (\text{Linearity of Trace})\\
			=& \sum_{i, j} \alpha_i \alpha_j \abs{\braket{\psi_i | W_x | \psi_j}}^2.
		\end{align*}
		Observe that $0 \leq \abs{\braket{\psi_i | W_x | \psi_j}} \leq 1$, completing the proof.
	\end{proof}
	
	\begin{fact}\label{fact:p-hat-sum-1}
		For mixed state $\rho$,
		\[
		2^n \sum_{x \in \F_2^{2n}} \hat{p}_\rho(x) = 1.
		\]
	\end{fact}
	\begin{proof}
		Using \cref{fact:duality} with $T = \{0^{2n}\}$, we find that 
		\[2^n \sum_{x \in \F_2^{2n}} \hat{p}_\rho(x) = 2^n p_\rho(0^{2n}) = 1. \qedhere\]
	\end{proof}
	
	\subsection{Stabilizer States}
	The Clifford unitaries are any unitary that can be generated by Hadamard, Phase, and CNOT gates.
	The set of stabilizer states are any state that can be generated by Clifford unitaries from the all zeros state.
	A state $\rho$ is a stabilizer state if and only if the set \[\weyl(\rho) \coloneqq \{x \in \F_2^{2n} : \trace(W_x \rho) = \pm 1\}\] forms a Lagrangian subspace.
	Note that $\weyl(\rho)$ always forms an isotropic subspace, so the difference is whether or not it is Lagrangian as well.
	
	\begin{definition}[Stabilizer Fidelity]
		For any state $\rho$, the stabilizer fidelity is
		\[
		\calF(\rho) = \max_{\ket \phi}\braket{\phi | \rho | \phi}
		\]
		where the maximization is over all stabilzer states $\ket \phi$.
	\end{definition}
	
	Note that for any Clifford unitary $C$, $\calF(C \rho C^\dagger) = \calF(\rho)$, which we will use to simplify proofs.
	Observe also that $p_\rho(x) = p_{C\rho C^\dagger}(y)$ for $y \in \F_2^{2n}$ such that $W_y = \pm C W_x C^\dagger$.\footnote{The Clifford group is well-known to normalize the Pauli group.}

	\section{The Algorithm}
	
	Like many stabilizer state algorithms, we use Bell different sampling as a starting point \cite{gross2021schur}.
	
	\begin{definition}[Bell Difference Sampling]\label{def:bell-diff}
		Given $4$ copies of a state $\rho$ there exists a time algorithm to measure from the distribution $q_\rho$ defined as
		\[
		q_\rho(a) \coloneqq \frac{1}{4^n} \sum_{x \in \F_2^{2n}} (-1)^{[a, x]} \trace^4(W_x \rho).
		\]
		The algorithm uses only two-copy Clifford measurements and no ancilla qubits, and it runs in linear time.
	\end{definition}
	
	The symplectic Fourier decomposition neatly falls out of this identity.
	\begin{fact}\label{fact:q-hat}
		For all $x \in \F_2^{2n}$ and mixed state $\rho$:
		\[
		\hat{q}_\rho(x) = \frac{\trace^4(W_x \rho)}{4^n} = p_\rho(x)^2.\]
	\end{fact}
	
	\begin{remark}
		Using \cref{thm:convolution-theorem}, one can see that $q_\rho = 16^n (\hat{p}_\rho \ast \hat{p}_\rho)$.
		Like with $p_\Psi$, it is unclear how to sample from $2^n \hat{p}_\rho$ without extra power (such as the complex conjugate state).
		However, Bell sampling \cite{montanaro-bell-sampling} (without the difference/convolution) actually samples from a related distribution $t_\rho$ whose symplectic Fourier coefficients are proportional to $p_\rho$ up to sign (i.e., $\hat{t}_\rho \propto \pm p_\rho$) \cite{damanik2018optimality}.
		The convolution squares the Fourier coefficients, thereby removing the $\pm 1$ phase, leaving us with $q_\rho$ as desired.
	\end{remark}

	Most algorithms on stabilizer property testing \cite{grewal_et_al:LIPIcs.ITCS.2023.64,grewal2023improved,arunachalam2024toleranttestingstabilizerstates,arunachalam2024notepolynomialtimetoleranttesting,bao2024toleranttestingstabilizerstates,mehraban2024improvedboundstestinglow} derive from that of \cite{gross2021schur}.
	This test accepts with probability \[\frac{1 + 2^n \sum_{x \in \F_2^{2n}} q_\rho(x) p_\rho(x)}{2}.\]
	Let $\eta_{\text{GNW}} \coloneqq 2^n \sum_{x \in \F_2^{2n}} q_\rho(x) p_\rho(x)$ be the bias of this test.
	When $\rho$ is pure, the use of \cref{fact:plancherel} and \cref{remark:p-hat-pure} show that this becomes $\eta_{\text{GNW}} = 4^n \sum_{x \in \F_2^{2n}} p_\rho(x)^3$, which is highly correlated to the Gowers norm \cite{arunachalam2024toleranttestingstabilizerstates}.
	
	Since $\hat{p}_\rho \neq \frac{p_\rho}{2^n}$ in general, we need to revise our test.
	The key is that we want to have a bias of \[
	\eta \coloneqq 4^n \sum_{x \in \F_2^{2n}} q_\rho(x) \hat{p}_\rho(x),
	\]
	which will become $4^n \sum_{x \in \F_2^{2n}} p_\rho(x)^3$ after using \cref{fact:plancherel}, \cref{fact:q-hat}, and \cref{fact:double-hat}.

	\begin{proposition}\label{prop:new-test}
		For all mixed states $\rho$
		\[\sum_{x \in \F_2^{2n}} q_\rho(x) \hat{p}_\rho(x) = \sum_{x \in \F_2^{2n}} p_\rho(x)^3.\]
	\end{proposition}
	\begin{proof}
		\begin{align*}
			\sum_{x \in \F_2^{2n}} q_\rho(x) \hat{p}_\rho(x)
			&= 4^n \sum_{x \in \F_2^{2n}} \hat{q}_\rho(x) \hat{\hat{p}}_\rho(x) && (\mathrm{\cref{fact:plancherel}})\\
			&= \sum_{x \in \F_2^{2n}} p_\rho(x)^3 && (\mathrm{\cref{fact:q-hat}, \cref{fact:double-hat}}) \qedhere
		\end{align*}
	\end{proof}
	
	To actually estimate this quantity, we run an ancilla-free SWAP test between $\rho$ and $W_x \rho W_x$.

	\begin{lemma}[Ancilla Free SWAP Test]\label{lem:swap-test}
		Given two states, $\rho$ and $\sigma$, there exists a measurement that outputs a $\pm 1$ random variable whose mean is $\trace(\rho \sigma)$ and uses one Bell measurement and no ancilla qubits in linear time.
	\end{lemma}
	\begin{proof}
		Define the Bell basis as $\{ W_{0^{2n}} \otimes W_a \ket{\Phi} : a \in \F_2^{2n}\}$ where $\ket \Phi \coloneqq \frac{1}{\sqrt{2^n}} \sum_{x \in \F_2^n} \ket{x} \otimes \ket{x}$. The Bell basis vectors form a complete and orthonormal basis.
		
		Let $\mathrm{SWAP}$ be the $2n$ qubit gate that swaps the first $n$ qubits with the last $n$ qubits.
		It can be shown (see e.g., \cite[Eqn S46]{hangleiter2023bell}) that for all $a = (v, w) \in \F_2^{2n}$ where $v \in \F_2^n$ is the first $n$ bits and $w \in \F_2^n$ the last $n$ that
		\[
		\mathrm{SWAP} \cdot \left(W_{0^{2n}} \otimes W_a\right) \ket{\Phi} =  (-1)^{v \cdot w} W_{0^{2n}} \otimes W_a \ket{\Phi}.
		\]
		That is, $\mathrm{SWAP}$ is diagonal in the Bell basis with $\pm 1$ eigenvalues.
		By measuring $\rho \otimes \sigma$ in the Bell basis, we are returned the Bell state $W_{0^{2n}} \otimes W_a \ket{\Phi}$.
		If we let $a = (v, w) \in \F_2^{2n}$, then outputting $(-1)^{v \cdot w}$ lets us efficiently estimate
		\[
		\trace(\mathrm{SWAP} \cdot \rho \otimes \sigma) = \trace(\rho \sigma),
		\]
		using another well-known identity (see e.g., \cite[Eqn S47]{hangleiter2023bell}).
	\end{proof}
	
	\begin{corollary}[$\hat{p}_\rho$ estimator]\label{cor:p-hat-estimator}
		Given two copies of $\rho$ and $x \in \F_2^{2n}$, there exists a measurement that outputs a $\pm 1$ random variable whose mean is \[4^n \hat{p}_\rho(x)\] and uses one Bell measurement, two $W_x$ gates, and no ancilla qubits in linear time.
	\end{corollary}
	\begin{proof}
		Note that all $W_x$ are unitaries.
		Therefore, run the \nameref{lem:swap-test} between $\rho$ and $W_x \rho W_x$ to get a mean of $\trace\left(\rho \left(W_x \rho W_x\right)\right) = 4^n \hat{p}_\rho(x)$.
	\end{proof}
	
	We can now complete our tester by merging Bell difference sampling with \cref{cor:p-hat-estimator}.
	That is, we will sample $x \in \F_2^{2n}$ via Bell difference sampling, then creates an unbiased estimator for $\hat{p}_\rho(x)$ using the \nameref{lem:swap-test}.
	
	\begin{theorem}\label{thm:alg}
		Given $6$ copies of $\rho$, there exists an algorithm that outputs a $\pm 1$ random variable whose mean is \[4^n \sum_{x \in \F_2^{2n}} p_\rho(x)^3.\]
		The algorithm uses only two copies at a time, requires the Clifford operations, no ancilla qubits, and runs in linear time.
	\end{theorem}
	\begin{proof}
		Using \cref{def:bell-diff}, we can use $4$ copies to measure $x \in \F_2^{2n}$ from $q_\rho$.
		We then use \cref{cor:p-hat-estimator} to estimate $4^n \hat{p}_\rho$, such that by linearity of expectations the mean of our random variable is
		\[
		\eta \coloneqq 4^n\Ex_{x \sim q_\rho}\left[\hat{p}_\rho(x)\right] = 4^n \sum_{x \in \F_2^{2n}} q_\rho(x) \hat{p}_\rho(x) = 4^n \sum_{x \in \F_2^{2n}} p_\rho(x)^3
		\]
		by \cref{prop:new-test}.
		Since both Bell difference sampling and \cref{cor:p-hat-estimator} run in linear time, so does our final test.
		Finally, we observe that Bell measurements can be performed using Clifford operations only.
	\end{proof}
	
	\begin{remark}\label{remark:alternative-measure}
		A test that measures $\eta^\prime \coloneqq 32^n \sum_{x \in \F_2^{2n}} \hat{p}_\rho(x)^3$ would have also sufficed (though incurring a quadratic loss), as \[
		\eta \geq \eta^\prime \geq \eta^2.
		\]
		However, it is not clear how to sample from a distribution $q^\prime_\rho$ such that $\hat{q^\prime}_\rho(a) \propto \hat{p}_\rho(a)^2$ (i.e., $q^\prime_\rho \propto (p_\rho \ast p_\rho)$), since $p_\rho$ is no longer a distribution.
	\end{remark}
	
	\section{Completeness for Mixed States}
	We now generalize the completeness analysis of \cite{grewal_et_al:LIPIcs.ITCS.2023.64, grewal2023improved} to mixed state inputs.
	In \cref{lem:p-lower} we show that the sum of the $p_\rho$-mass on any Lagrangian subspace $L$ is lower bounded by the fidelity with any stabilizer state $\ket \phi$ such that $\weyl(\ket \phi) = L$.
	Using H\"older's inequality and the non-negativitiy of $p_\rho$, in \cref{cor:eta-lower} we can extend this to lower bound the quantity measured in \cref{thm:alg} by $\calF(\rho)^6$.
	
	\begin{lemma}\label{lem:p-lower}
		Given an $n$-qubit state $\rho$ and stabilizer state $\ket \phi$ such that $L = \weyl(\ketbra{\phi}{\phi})$,
		\[\sum_{x \in L} p_\rho(x) \geq \braket{\phi | \rho | \phi}^2.\]
	\end{lemma}
	\begin{proof}
		For simplicity, we will take the Clifford unitary $C$ that maps $\ket \phi$ to the all zeros state.
		Then $\sum_{x \in L} p_\rho(x)$ is now equal to the sum of $p_{C \rho C^\dagger}$ over $L^\prime = 0^n \otimes \F_2^n = \weyl(\ketbra{0^n}{0^n})$.
		Using Cauchy-Schwarz we find that
		\begin{align*}
			\sum_{x \in 0^n \times \F_2^n} p_{C \rho C^\dagger}(x)
			&= \frac{1}{2^n}\sum_{x \in 0^n \times \F_2^n} \trace^2\left(C \rho C^\dagger W_x\right)\\
			&\geq \frac{1}{4^n} \left(\sum_{x \in 0^n \times \F_2^n} \trace\left(C \rho C^\dagger W_x\right)\right)^2\\
			&= \frac{1}{4^n} \left(2^n \trace\left(C \rho C^\dagger \ketbra{0^n}{0^n}\right)\right)^2\\
			&= \braket{\phi | \rho | \phi}^2. \qedhere
		\end{align*}
	\end{proof}

	\begin{corollary}\label{cor:eta-lower}
		Given an $n$-qubit state $\rho$,
		\[4^n \sum_{x \in \F_2^{2n}} p_\rho(x)^3 \geq \calF(\rho)^6.\]
	\end{corollary}
	\begin{proof}
		Let $\ket \phi$ be the stabilizer state that achieves maximum fidelity $\rho$ and let $L = \weyl(\ket \phi)$.
		\begin{align*}
			\calF(\rho) &= \braket{\phi | \rho | \phi}\\
			&\leq \sqrt{\sum_{x \in L} p_\rho(x)} && (\mathrm{\cref{lem:p-lower}})\\
			& \leq \left(4^n \sum_{x \in L} p_\rho(x)^3\right)^{1/6} && (\text{H\"olders inequality})\\
			&\leq \left(4^n \sum_{x \in \F_2^{2n}} p_\rho(x)^3\right)^{1/6}. && (p_\rho(x) \geq 0) \qedhere
		\end{align*}
	\end{proof}

	\section{Soundness for Mixed States}
	To show soundness of the test, previous work for pure states used the fact that $\calF(\ketbra{\psi}{\psi}) \geq \sum_{x \in L} p_{\Psi}(x)$ for any Lagrangian subspace $L \subset \F_2^{2n}$.
	Thus, the soundness proof would proceed by showing that if the test accepted with high probability, then there exists a Lagrangian subspace $L$ such that $\sum_{x \in L} p_{\Psi}(x)$ was large.
	While this bound remains true for mixed states (see \cref{lem:fidelity-lower}), we will use $\hat{p}_\rho$ as the lower bounding tool due to its relationship to \cref{thm:alg}.
	
	We first show why using $\hat{p}_\rho$ is sufficient, starting with the lower bound using $p_\rho$, then using \cref{fact:duality} to convert it into a bound for $\hat{p}_\rho$.
	
	\begin{lemma}[Corollary of {\cite[Lemma B.9]{grewal2023efficient}}]\label{lem:fidelity-lower}
		For an arbitrary mixed state $\rho$ and Lagrangian subspace $L$, there exists a stabilizer state $\ket \phi$ with $\weyl(\ket \phi) = L$ such that \[\braket{\phi | \rho | \phi} \geq \sum_{x \in L} p_\rho(x).\]
	\end{lemma}

	\begin{corollary}\label{cor:fidelity-lower}
		For arbitrary mixed state $\rho$ and Lagrangian subspace $L$:
		\[\calF(\rho) \geq 2^n \sum_{x \in L} \hat{p}_\rho(x)\]
	\end{corollary}
	\begin{proof}
		Using \cref{lem:fidelity-lower}, by definition we see that \[\calF(\rho) \coloneqq \max_{\ket \phi^\prime} \braket{\phi^\prime | \rho | \phi^\prime} \geq \braket{\phi | \rho | \phi} \geq \sum_{x \in L} p_\rho(x).\]
		Since $L$ is Lagrangian, $L^\sympcomp = L$.
		Using \cref{fact:duality}, we find that
		\[
		\sum_{x \in L} p_\rho(x) = 2^n \sum_{x \in L} \hat{p}_\rho(x),
		\]
		thus completing the proof.
	\end{proof}
	
	We now show that $4^n \sum_{x \in \F_2^{2n}} p_\rho(x)^3$ measures `linearity' of $\hat{p}_\rho$ in a sense.
	
	\begin{lemma}\label{lem:linearity}
		\[
		4^n \sum_{x \in \F_2^{2n}} p_\rho(x)^3 = 16^{n}\sum_{x, y \in \F_2^{2n}}\hat{p}_\rho(x) \hat{p}_\rho(y) \hat{p}_\rho(x+y) .
		\]
	\end{lemma}
	\begin{proof}
		Via linearity of expectation,
		\begin{align*}
			\sum_{x, y \in \F_2^{2n}}\hat{p}_\rho(x) \hat{p}_\rho(y) \hat{p}_\rho(x+y)
			&= 4^n \sum_{x \in \F_2^{2n}} \left[\hat{p}_\rho(x) (\hat{p}_\rho \ast \hat{p}_\rho)(x)\right]\\
			&= 16^n \sum_{x \in \F_2^{2n}} \hat{\hat{p}}_\rho(x) \wh{\hat{p}_\rho \ast \hat{p}_\rho}(x) && (\mathrm{\cref{fact:plancherel}})\\
			&= \frac{1}{4^{n}}\sum_{x \in \F_2^{2n}} p_\rho(x) p_\rho(x)^2. && (\mathrm{\cref{thm:convolution-theorem}, \cref{fact:double-hat}}) \qedhere\\
		\end{align*}
	\end{proof}
	
	The basic steps of soundness will be to: 1) find an approximate subspace $S$ where every entry has large $\hat{p}_\rho$ values, 2) find an affine subspace $z+V$ that has large intersection with $S$ and therefore has large $\hat{p}_\rho$ mass, 3) argue that the proper subspace $V$ has just as much $\hat{p}_\rho$ mass as $z+V$, and 4) argue that some \emph{isotropic} subspace has large intersection with $V$ and thus has large $\hat{p}_\rho$ mass.
	To summarize, the path we take is
	\[
	\text{approximate subspace } \rightarrow \text{ affine subspace } \rightarrow \text{ proper subspace } \rightarrow \text{ isotropic subspace}
	\]
	where at each step we argue that the $\hat{p}_\rho$ mass only goes down by at most a polynomial factor.
	
	We now generalize the proof ideas from \cite{arunachalam2024toleranttestingstabilizerstates,arunachalam2024notepolynomialtimetoleranttesting,bao2024toleranttestingstabilizerstates} to show that there exists an approximate subgroup of $\F_2^{2n}$ with large $\hat{p}_\rho$ mass.
	This is because the major properties of $p_\Psi$ used are 1) $\sum_{x \in \F_2^{2n}} p_\Psi(x) \leq 1$, 2) $0 \leq p_\Psi(x) \leq \frac{1}{2^n}$, and 3) $\hat{p}_\Psi \geq 0$, which are satisfied by $2^n \hat{p}_\rho$ as well.
	Because the following proofs essentially mirror that of prior work, we omit them where applicable and only give formal proofs where the proof breaks down for mixed states.

	\begin{restatable}[Generalization of {\cite[Theorem 4.7]{arunachalam2024toleranttestingstabilizerstates}}]{lemma}{approxsubspacerestate}\label{thm:approx-subspace}
		Let $f : \F_2^{2n} \rightarrow \mathbb{R}^{\geq 0}$ be an arbitrary function such that $0 \leq f(x) \leq \frac{1}{2^n}$ and $\sum_{x \in \F_2^{2n}} f(x) \leq 1$.
		If \[2^n \sum_{x, y \in \F_2^{2n}} f(x) f(y) f(x+y) \geq \gamma\]
		then there exists a \emph{subset} $S \subseteq \F_2^{2n}$ of size $\abs{S} \in [\gamma^2/80, 4/\gamma] \cdot 2^n$ satisfying
		\begin{itemize}
			\item $\forall x \in S$, $2^n f(x) \geq \gamma/4$
			\item $\Pr_{x, y \in S}\left[x + y \in S\right] \geq \Omega(\gamma^5)$.
		\end{itemize}
	\end{restatable}
	\begin{proof}[Proof Sketch]
		The only step that uses the fact that $p_\Psi$ is a distribution is \cite[Fact 2.1]{arunachalam2024toleranttestingstabilizerstates}, which is used to show that \[M \coloneqq \{x \in \F_2^{2n} : 2^n p_\Psi(x) \geq \frac{\gamma}{4}\}\]
		has size at least $\abs{M} \geq \frac{\gamma}{2} 2^n$.
		We now formally prove that statement for $f$, with an improvement in the constant.
		
		\vspace{-2em}
		\begin{changemargin}{1cm}{1cm} 
			\begin{claim}\label{claim:many-large-p-hat}
				Let $M \coloneqq \{x \in \F_2^{2n} : 2^n f(x) \geq \frac{\gamma}{4}\}$, then $\abs{M} \geq \frac{3}{4}\gamma 2^n$.
			\end{claim}
			\begin{proof}
				We first show the following lower bound using Cauchy-Schwarz
				\[
				\gamma \leq 2^n \sum_{x, y \in \F_2^{2n}} f(x)f(y)f(x+y) \leq 2^n \sum_{x, y \in \F_2^{2n}}f(x)^2 f(y) \leq 2^n \sum_{x \in \F_2^{2n}} f(x)^2.
				\]
				
				Observe that $\sum_{x \in M} f(x) + \sum_{x \not\in M} f(x) \leq 1$ such that $2^n \max_{x \in M} f(x) \leq 1$ and $2^n \max_{x \not\in M} f(x) < \frac{\gamma}{4}$.
				Therefore
				\[
				2^n \sum_{x \in \F_2^{2n}} f(x)^2 = 2^n \left(\sum_{x \in M} + \sum_{x \not\in M}\right) f(x)^2 \leq \sum_{x \in M} f(x) + \frac{\gamma}{4} \sum_{x \not \in M} f(x) \leq (1 - \frac{\gamma}{4}) \sum_{x \in M} f(x) + \frac{\gamma}{4}.
				\]
				We thus have $\gamma \leq (1 - \frac{\gamma}{4}) \sum_{x \in M} f(x) + \frac{\gamma}{4}$, which implies $\sum_{x \in M} f(x) \geq \frac{3 \gamma}{4-\gamma} \geq \frac{3}{4}\gamma$.
				Noting that $f(x) \leq \frac{1}{2^n}$ completes the proof.
			\end{proof}
		\end{changemargin}

		The rest of the proof follows without change.
		Since $\sum_{x \in \F_2^{2n}} f(x) \leq 1$, $\abs{M} \leq \frac{4}{\gamma}2^n$.
		We probabilistically create the set $S$ such that $x \in M$ is added to $S$ with probability $2^n f(x)$.
		Since $S \subseteq M$, $2^n f(x) \geq \frac{\gamma}{4}$ for all $x \in S$.
		Taking the expectation over $S$, we find that  $\Ex_{S}\left[\Pr_{x, y \in S}\left[x + y \in S\right]\right] \geq \Omega(\gamma^5)$, by relating it to $2^n \sum_{x, y \in \F_2^{2n}} f(x) f(y) f(x+y)$.\footnote{This is the most involved step of the proof. See \cite[Proof of Theorem 4.5]{arunachalam2024toleranttestingstabilizerstates} for details.}
		Furthermore, $\Ex_{S}\left[\abs{S}\right] \geq \frac{\gamma}{4} \abs{M} \geq \frac{3}{16}\gamma^2 2^n$.
		Using a Chernoff bound, we find that there is a non-zero probability of $S$ satisfying both $\abs{S} \geq \frac{\gamma^2}{80} 2^n$ and $\Pr_{x, y \in S}\left[x + y \in S\right] \geq \Omega(\gamma^5)$ simultaneously.
		This implies that such an $S$ must exist, by the probabilistic method.
		Noting that for all possible $S$, $\abs{S} \leq \abs{M} \leq \frac{4}{\gamma}2^n$ allows us to satisfy the final condition.
	\end{proof}
	
	\begin{remark}
		It is possible to use \cite[Theorem 4.7]{arunachalam2024toleranttestingstabilizerstates} directly, but this would have incurred an extra factor $2$ in the exponent of $\poly(\gamma)$ of the final result.
		See \cref{remark:alternative-measure} for reasoning.
		It also (minorly) complicates the later proofs, in that we would have to use \cref{cor:symplectic-subspace-upper} instead of \cref{lem:symplectic-subspace-upper}.
	\end{remark}
	
	We now want to show that we can turn $S$  from \cref{thm:approx-subspace} into an actual subspace.
	To start, we turn $S$ into an \emph{affine} subspace/coset $z+V$.
	The following was (implicitly) derived in \cite{arunachalam2024toleranttestingstabilizerstates} using results from additive combinatorics.
	
	\begin{lemma}[{\cite[Proof of Corollary 4.11]{arunachalam2024toleranttestingstabilizerstates}}]\label{lem:additive-comb}
		Let $S \subseteq \F_2^{n}$ with $\Pr_{x, y \in S}\left[x + y \in S \right] \geq \eps$. Then there exists a subset $S^\prime \subseteq S$ and subspace $V \subseteq \F_2^n$ and $z \in \F_2^n$ (i.e., an affine subspace $z+V$) with $\abs{S^\prime \cap (z + V)} \geq  O(\eps^{72}) \abs{S^\prime}$, $\abs{V} \leq \abs{S^\prime}$, and $\abs{S^\prime} \geq \frac{\eps}{3}\abs{S}$.
	\end{lemma}
	
	\begin{restatable}[Generalization of {\cite[Corollary 4.11]{arunachalam2024toleranttestingstabilizerstates}}]{corollary}{actualsubspacerestate}\label{thm:actual-subspace}
		Let $f : \F_2^{2n} \rightarrow \mathbb{R}^{\geq 0}$ be an arbitrary function such that $0 \leq f(x) \leq \frac{1}{2^n}$ and $\sum_{x \in \F_2^{2n}} f(x) \leq 1$.
		If \[2^n \sum_{x, y \in \F_2^{2n}} f(x) f(y) f(x+y) \geq \gamma\]
		then there there exists a subspace $V \subseteq \F_2^{2n}$ and $z \in \F_2^{2n}$ (i.e., an affine subspace $z+V$) such that $\frac{2^n}{\abs{V}}\sum_{x \in V} f(x+z) \geq \Omega(\gamma^{361})$ and $\abs{V} \geq \Omega(\gamma^{367}) 2^n$.
	\end{restatable}
	\begin{proof}
		First apply \cref{thm:approx-subspace} to get an approximate subspace $S$ such that for all $x \in S$, $2^n f(x) \geq \frac{\gamma}{4}$.
		Use \cref{lem:additive-comb} with $\eps = \Omega(\gamma^5)$ to find a subset $S^\prime \subseteq S$ and subspace $V \subseteq \F_2^{2n}$ and $z \in \F_2^{2n}$ such that $z + V$ shares $O(\gamma^{360})\abs{S^\prime}$ elements with $S^\prime$ and $\abs{S^\prime} \geq O(\gamma^5) \abs{S}$. Since $\abs{S} \geq \Omega(\gamma^2 )2^n$,
		\[\abs{V} \geq \abs{S^\prime \cap (z+V)} \geq O(\gamma^{360})\abs{S^\prime} \geq O(\gamma^{365})\abs{S} \geq \Omega(\gamma^{367}) 2^n.\]
		Furthermore, since each element of $S$, and therefore $S^\prime$, has large $f(x)$, we get that 
		\[
		\frac{2^n}{\abs{V}}\sum_{x \in V} f(x+z) \geq \frac{2^n}{\abs{S^\prime}}\sum_{x \in z+V} f(x) \geq \frac{2^n}{\abs{S^\prime}}\sum_{x \in S^\prime \cap (z + V)} f(x) \geq 2^n \frac{\abs{S^\prime \cap (z+V)}}{\abs{S^\prime}} \frac{\gamma}{4 \cdot 2^n} \geq \Omega(\gamma^{361}). \qedhere
		\]
	\end{proof}

	To turn this affine subspace into a proper subspace, we note that undoing the affine shift cannot reduce the sum (i.e, $\sum_{x \in V} \hat{p}(x+z) \leq \sum_{x \in V} \hat{p}(x)$).
	The proof of \cite[Claim 4.13]{arunachalam2024toleranttestingstabilizerstates}, which performs this step, uses the fact that $\hat{p}_\Psi = \frac{p_\Psi}{2^n}$ and thus does not immediately extend to mixed states. 
	We replace both it and \cite[Claim 4.12]{arunachalam2024toleranttestingstabilizerstates} with a simple proof that does not rely on this identity.
	Instead, it only uses the fact that the Fourier coefficients of $\hat{p}_\rho$ are non-negative.
	As a consequence, this property also immediately extends to $p_\rho$ as well.
	
	\begin{lemma}\label{lem:proper-subspace-more-mass}
		Let $f : \F_2^{2n} \rightarrow \mathbb{R}$ such that $\hat{f}(a) \geq 0$ for all $a \in \F_2^{2n}$.
		Then for any subspace $V \subseteq \F_2^{2n}$ and $z \in \F_2^{2n}$
		\[
		\sum_{x \in V} f(x) \geq \sum_{x \in V} f(x+z).
		\]
	\end{lemma}
	\begin{proof}
		If $z \in V$ then the statement is trivial, so assume now that $z \not\in V$.
		Let $U = \langle z \rangle + V$ be the subspace created by adding $z$ as a generator to $V$.
		Note that $\abs{U} = 2 \abs{V}$ and \[\sum_{x \in V} f(x + z) = \sum_{x \in U} f(x) - \sum_{x \in V} f(x).\]
		Because $V \subset U$ we see that $U^\sympcomp \subset V^\sympcomp$.
		Since $\hat{f}(x) \geq 0$, by \cref{fact:duality}
		\[
		0 \leq 2\abs{V}\left(\sum_{x \in V^\sympcomp} - \sum_{x \in U^\sympcomp}\right) \hat{f}(x) = \left( 2\sum_{x \in V} -  \sum_{x \in U} \right) f(x) = \sum_{x \in V} f(x) - \sum_{x \in V} f(x+z). \qedhere
		\]
	\end{proof}
	
	Finally, we show that there must be a large isotropic subspace within $V$.
	This will (eventually) allow us to appy \cref{cor:fidelity-lower}.
	
	The following is implicit in the proof of \cite[Fact 2.5]{arunachalam2024notepolynomialtimetoleranttesting} and {\cite[Fact 4.16]{arunachalam2024toleranttestingstabilizerstates}}.
	They implicitly show that \[2^n \sum_{x \in A} p_\Psi(x) \leq \sqrt{\abs{A}}\] for symplectic subspace $A$
	and then immediately use the fact that $p_\Psi = \frac{\hat{p}_\Psi}{2^n}$ in the proof.
	They then proceeds to analyze $4^n \sum_{x \in A} \hat{p}_\Psi$ using the fact that the Pauli group forms a 1-unitary design.
	Since we are starting with $\hat{p}_\rho$ in \cref{lem:symplectic-subspace-upper}, this conversion step is completely avoided and the following holds.
	
	\begin{lemma}[{\cite[Proof of Fact 4.16]{arunachalam2024toleranttestingstabilizerstates}}]\label{lem:symplectic-subspace-upper}
		For subspace $A \subseteq \F_2^{2n}$ such that $A \cap A^\sympcomp = \{0^{2n}\}$ (i.e., a symplectic subspace)
		\[
		4^n \sum_{x \in A} \hat{p}_\rho(x) \leq \sqrt{\abs{A}}.
		\]
	\end{lemma}

	For completeness sake, we also give a proof that holds for $p_\rho$.
	
	\begin{corollary}\label{cor:symplectic-subspace-upper}
		For subspace $A \subseteq \F_2^{2n}$ such that $A \cap A^\sympcomp = \{0^{2n}\}$ (i.e., a symplectic subspace)
		\[
		2^n \sum_{x \in A} p_\rho(x) \leq \sqrt{\abs{A}}.
		\]
	\end{corollary}
	\begin{proof}
		Note that $A^\sympcomp$ is also a symplectic subspace and that $\abs{A} \cdot \abs{A^\sympcomp} = 4^n$.
		Using \cref{fact:duality} followed by \cref{lem:symplectic-subspace-upper}, we find that
		\[
		2^n \sum_{x \in A} p_\rho(x) = 2^n \abs{A} \sum_{x \in A^\sympcomp} \hat{p}_\rho(x) \leq \frac{\abs{A} \sqrt{\abs{A^\sympcomp}}}{2^n} = \sqrt{\abs{A}}. \qedhere
		\]
	\end{proof}
	
	We now complete the proof of soundness by using \emph{mutually unbiased basis}, which was inspired by \cite{bao2024toleranttestingstabilizerstates}.
	
	\begin{fact}[{\cite{Bandyopadhyay2002mub}}]\label{fact:mutually-unbiased-basis}
		Every symplectic subspace $S$ of dimension $2k$ can be covered via $2^k + 1$ isotropic subspaces $T_1, \dots, T_{2k+1} \subset S$. These isotropic subspaces are all of dimension $k$ and thus only have the trivial intersection.
	\end{fact}
	
	\begin{restatable}{theorem}{stabilizercover}\label{thm:stabilizer-cover}
		\[\text{If } 4^n \sum_{x \in \F_2^{2n}} p_\rho(x)^3 \geq \gamma, \text{ then } \calF(\rho) \geq \Omega(\gamma^{1089}).\]
	\end{restatable}
	\begin{proof}
		By \cref{lem:linearity}, the fact that $2^n \hat{p}$ forms a distribution, and \cref{thm:actual-subspace}, there exists a subspace $V \subseteq \F_2^{2n}$ and $z \in \F_2^{2n}$ such that $\frac{4^n}{\abs{V}} \sum_{x \in V} \hat{p}(x+z) \geq \Omega(\gamma^{361})$ and $\abs{V} \geq \Omega(\gamma^{367}) 2^n$.
		By \cref{lem:proper-subspace-more-mass}, \cref{fact:double-hat} and the fact that $p(x) \geq 0$ we see that, $4^n \sum_{x \in V} \hat{p}(x) \geq \Omega(\gamma^{329}) \abs{V}$ as well.

		Via symplectic Gram-Schmidt \cite{fattal2004entanglementstabilizerformalism,grewal2024pseudoentanglementaintcheap}, it is well-known that every subspace $V \in \F_2^{2n}$ of dimension $d$ has generators of the form \[x_1, z_1, x_2, z_2, \dots, x_k, z_k, z_{k+1}, \dots, z_{d - 2k} \in \F_2^{2n}\] such that \[[x_i, z_j] = \delta_{ij} \text{ and } [x_i, x_j] = 0 = [z_i, z_j].\]
		Let $S$ the symplectic subspace generated by $x_1, z_1, \dots, x_k, z_k$ and let $T_1, \dots T_{2k+1} \subset S$ be the mutually unbiased basis of $S$ from \cref{fact:mutually-unbiased-basis}.
		Taking the generators of $T_i$ and combining it with $\{z_k, \dots z_{d-2k}\}$ naturally generates an isotropic subspace $T_i^\prime \subset V$ such that $\bigcup_{i=1}^{2k+1} T_i^\prime = V$.
		So if we can show that $\abs{S} = 4^k$ is small then, by an averaging argument, there exists an isotropic subspace $T$ such that \[2^n \sum_{x \in T} \hat{p}_\rho(x) \geq \Omega(\gamma^{361}) \frac{\abs{V}}{(2^k + 1) 2^n} \geq \frac{\Omega(\gamma^{728})}{\sqrt{\abs{S}} + 1}.\]
		Using \cref{cor:fidelity-lower} and the fact that $\hat{p}_\rho(x) \geq 0$ (i.e., we can extend $T$ to some arbitrary Lagrangian subspace that contains $T$) would complete the proof.
		
		To actually upper bound $\abs{S}$, we use \cref{lem:proper-subspace-more-mass} and \cref{lem:symplectic-subspace-upper} to show that, \[O(\gamma^{361}) \abs{V} \leq 4^n \sum_{x \in V} \hat{p}_\rho(x) =  4^n\sum_{x\in S}\sum_{i=0}^{\frac{\abs{V}}{\abs{S}} - 1} \hat{p}_\rho(x + z_{i + k}) \leq 4^n \frac{\abs{V}}{\abs{S}}\sum_{x\in S}\hat{p}_\rho(x) \leq \sqrt{S} \frac{\abs{V}}{\abs{S}} = \frac{\abs{V}}{\sqrt{\abs{S}}}\]
		so $\sqrt{\abs{S}} = O(\gamma^{-361})$. Thus $\calF(\rho) \geq \Omega(\gamma^{1089})$.
	\end{proof}

	\begin{remark}
		By taking the union of $V$ and $z+V$ (rather  than just looking at $V$), one can arrive at a bigger subspace $V^\prime$ such that $4^n \sum_{x \in V^\prime} \hat{p}(x) \geq \poly(\gamma) \abs{V^\prime}$ and $\abs{V^\prime} = 2\abs{V} \geq \poly(\gamma) 2^n$.
		This does not change the final exponent in \cref{thm:stabilizer-cover}, but it does improve the constant hidden by the big O notation by a factor of $2$.
	\end{remark}
	
	\section{Conclusion}
	We have now shown how the quantity measured in our new test is both upper and lower bounded by the stabilizer fidelity.
	Please refer to \cite[Theorem 1]{bao2024toleranttestingstabilizerstates} for how to turn \cref{cor:eta-lower} and \cref{thm:stabilizer-cover} into a property testing algorithm, as they achieve similar bounds in the exponent.
	
	As noted earlier, \cref{thm:stabilizer-cover} follows the proofs of \cite{arunachalam2024toleranttestingstabilizerstates,arunachalam2024notepolynomialtimetoleranttesting,bao2024toleranttestingstabilizerstates}, whose bounds have since been improved by \cite{mehraban2024improvedboundstestinglow}.
	We leave it as an open problem if the proof of \cite[Theorem 1.1]{mehraban2024improvedboundstestinglow} can be made to work for mixed states, as well as matching any other future improvements.
	
	\section*{Acknowledgements}
	We thank Marcel Hinsche for making us aware of this `ancilla-free' version of the SWAP test and for other instructive feedback on prior work regarding stabilizer state property testing. We also thank Sid Jain for helpful comments.
	We thank Arkopal Dutt for making us aware of errors in a preliminary draft pertaining to \cref{lem:additive-comb}.
	This work was done [in part] while the authors were visiting the Simons Institute for the Theory of Computing. VI is supported by an NSF Graduate Research Fellowship.
	DL is supported by the US NSF award FET-2243659 and US NSF Award CCF-222413.
	
	\bibliographystyle{alphaurl}
	\bibliography{refs}

\newcommand{\etalchar}[1]{$^{#1}$}
\begin{thebibliography}{GIKL24b}

\bibitem[ABD24]{arunachalam2024notepolynomialtimetoleranttesting}
Srinivasan Arunachalam, Sergey Bravyi, and Arkopal Dutt.
\newblock A note on polynomial-time tolerant testing stabilizer states, 2024.
\newblock URL: \url{https://arxiv.org/abs/2410.22220}, \href {https://arxiv.org/abs/2410.22220} {\path{arXiv:2410.22220}}.

\bibitem[AD24]{arunachalam2024toleranttestingstabilizerstates}
Srinivasan Arunachalam and Arkopal Dutt.
\newblock Towards tolerant testing stabilizer states, 2024.
\newblock URL: \url{https://arxiv.org/abs/2408.06289}, \href {https://arxiv.org/abs/2408.06289} {\path{arXiv:2408.06289}}.

\bibitem[BBRV02]{Bandyopadhyay2002mub}
Somshubhro Bandyopadhyay, P.~Oscar Boykin, Vwani Roychowdhury, and Farrokh Vatan.
\newblock {A New Proof for the Existence of Mutually Unbiased Bases }.
\newblock {\em Algorithmica}, 34:512--528, Nov 2002.
\newblock \href {https://doi.org/10.1007/s00453-002-0980-7} {\path{doi:10.1007/s00453-002-0980-7}}.

\bibitem[BGJ23]{bu2023stabilizertestingmagicentropy}
Kaifeng Bu, Weichen Gu, and Arthur Jaffe.
\newblock Stabilizer testing and magic entropy, 2023.
\newblock URL: \url{https://arxiv.org/abs/2306.09292}, \href {https://arxiv.org/abs/2306.09292} {\path{arXiv:2306.09292}}.

\bibitem[BvDH24]{bao2024toleranttestingstabilizerstates}
Zongbo Bao, Philippe van Dordrecht, and Jonas Helsen.
\newblock Tolerant testing of stabilizer states with a polynomial gap via a generalized uncertainty relation, 2024.
\newblock URL: \url{https://arxiv.org/abs/2410.21811}, \href {https://arxiv.org/abs/2410.21811} {\path{arXiv:2410.21811}}.

\bibitem[CGYZ24]{chen2024stabilizerbootstrappingrecipeefficient}
Sitan Chen, Weiyuan Gong, Qi~Ye, and Zhihan Zhang.
\newblock Stabilizer bootstrapping: A recipe for efficient agnostic tomography and magic estimation, 2024.
\newblock URL: \url{https://arxiv.org/abs/2408.06967}, \href {https://arxiv.org/abs/2408.06967} {\path{arXiv:2408.06967}}.

\bibitem[Dam18]{damanik2018optimality}
Raja Oktovin~Parhasian Damanik.
\newblock {Optimality in Stabilizer Testing}.
\newblock {\em Master's Thesis, August}, 2018.
\newblock URL: \url{https://eprints.illc.uva.nl/id/eprint/1622/1/MoL-2018-09.text.pdf}.

\bibitem[FCY{\etalchar{+}}04]{fattal2004entanglementstabilizerformalism}
David Fattal, Toby~S. Cubitt, Yoshihisa Yamamoto, Sergey Bravyi, and Isaac~L. Chuang.
\newblock Entanglement in the stabilizer formalism, 2004.
\newblock \href {https://arxiv.org/abs/quant-ph/0406168} {\path{arXiv:quant-ph/0406168}}.

\bibitem[GIKL23a]{grewal2023efficient}
Sabee Grewal, Vishnu Iyer, William Kretschmer, and Daniel Liang.
\newblock {Efficient Learning of Quantum States Prepared With Few Non-{C}lifford Gates}, 2023.
\newblock \href {https://arxiv.org/abs/2305.13409} {\path{arXiv:2305.13409}}.

\bibitem[GIKL23b]{grewal_et_al:LIPIcs.ITCS.2023.64}
Sabee Grewal, Vishnu Iyer, William Kretschmer, and Daniel Liang.
\newblock {Low-Stabilizer-Complexity Quantum States Are Not Pseudorandom}.
\newblock In {\em 14th Innovations in Theoretical Computer Science Conference (ITCS 2023)}, volume 251 of {\em Leibniz International Proceedings in Informatics (LIPIcs)}, pages 64:1--64:20, 2023.
\newblock \href {https://doi.org/10.4230/LIPIcs.ITCS.2023.64} {\path{doi:10.4230/LIPIcs.ITCS.2023.64}}.

\bibitem[GIKL24a]{grewal2023improved}
Sabee Grewal, Vishnu Iyer, William Kretschmer, and Daniel Liang.
\newblock {Improved Stabilizer Estimation via Bell Difference Sampling}.
\newblock In {\em Proceedings of the 56th Annual ACM Symposium on Theory of Computing}, STOC 2024, page 1352–1363, New York, NY, USA, 2024. Association for Computing Machinery.
\newblock \href {https://doi.org/10.1145/3618260.3649738} {\path{doi:10.1145/3618260.3649738}}.

\bibitem[GIKL24b]{grewal2024pseudoentanglementaintcheap}
Sabee Grewal, Vishnu Iyer, William Kretschmer, and Daniel Liang.
\newblock Pseudoentanglement ain't cheap, 2024.
\newblock URL: \url{https://arxiv.org/abs/2404.00126}, \href {https://arxiv.org/abs/2404.00126} {\path{arXiv:2404.00126}}.

\bibitem[GNW21]{gross2021schur}
David Gross, Sepehr Nezami, and Michael Walter.
\newblock {Schur--Weyl duality for the Clifford group with applications: Property testing, a robust Hudson theorem, and de Finetti representations}.
\newblock {\em Communications in Mathematical Physics}, 385(3):1325--1393, 2021.
\newblock \href {https://doi.org/10.1007/s00220-021-04118-7} {\path{doi:10.1007/s00220-021-04118-7}}.

\bibitem[HG23]{hangleiter2023bell}
Dominik Hangleiter and Michael~J. Gullans.
\newblock Bell sampling from quantum circuits, 2023.
\newblock \href {https://arxiv.org/abs/2306.00083v1} {\path{arXiv:2306.00083v1}}.

\bibitem[HH24]{hinsche2024singlecopystabilizertesting}
Marcel Hinsche and Jonas Helsen.
\newblock Single-copy stabilizer testing, 2024.
\newblock URL: \url{https://arxiv.org/abs/2410.07986}, \href {https://arxiv.org/abs/2410.07986} {\path{arXiv:2410.07986}}.

\bibitem[HK23]{haug2023scalable}
Tobias Haug and M.~S. Kim.
\newblock {Scalable Measures of Magic Resource for Quantum Computers}.
\newblock {\em PRX Quantum}, 4(1):010301, 2023.
\newblock \href {https://doi.org/10.1103/PRXQuantum.4.010301} {\path{doi:10.1103/PRXQuantum.4.010301}}.

\bibitem[HLK24]{Haug2024efficient}
Tobias Haug, Soovin Lee, and M.~S. Kim.
\newblock Efficient quantum algorithms for stabilizer entropies.
\newblock {\em Phys. Rev. Lett.}, 132:240602, Jun 2024.
\newblock URL: \url{https://link.aps.org/doi/10.1103/PhysRevLett.132.240602}, \href {https://doi.org/10.1103/PhysRevLett.132.240602} {\path{doi:10.1103/PhysRevLett.132.240602}}.

\bibitem[Mon17]{montanaro-bell-sampling}
Ashley Montanaro.
\newblock {Learning stabilizer states by Bell sampling}, 2017.
\newblock \href {https://arxiv.org/abs/1707.04012} {\path{arXiv:1707.04012}}.

\bibitem[MT24]{mehraban2024improvedboundstestinglow}
Saeed Mehraban and Mehrdad Tahmasbi.
\newblock Improved bounds for testing low stabilizer complexity states, 2024.
\newblock URL: \url{https://arxiv.org/abs/2410.24202}, \href {https://arxiv.org/abs/2410.24202} {\path{arXiv:2410.24202}}.

\bibitem[MW16]{montanaro2016survey}
Ashley Montanaro and Ronald~{de} Wolf.
\newblock {\em {A Survey of Quantum Property Testing}}.
\newblock Number~7 in Graduate Surveys. Theory of Computing Library, 2016.
\newblock \href {https://doi.org/10.4086/toc.gs.2016.007} {\path{doi:10.4086/toc.gs.2016.007}}.

\bibitem[PRR06]{parnas2006tolerant}
Michal Parnas, Dana Ron, and Ronitt Rubinfeld.
\newblock Tolerant property testing and distance approximation.
\newblock {\em Journal of Computer and System Sciences}, 72(6):1012--1042, 2006.
\newblock \href {https://doi.org/10.1016/j.jcss.2006.03.002} {\path{doi:10.1016/j.jcss.2006.03.002}}.

\end{thebibliography}
	
	
	\appendix
	
	\section{Soundness for Mixed States in the Close Regime}\label{appendix:close}
	
	While the results of \cite{arunachalam2024toleranttestingstabilizerstates,arunachalam2024notepolynomialtimetoleranttesting,bao2024toleranttestingstabilizerstates,mehraban2024improvedboundstestinglow} are all conceptually nice, due to their (unoptimized) constants, the bounds are largely useless from a practical standpoint and should generally not be used, as is. This is despite the previous state-of-the-art bounds from \cite{gross2021schur,grewal2023improved} of $\calF(\ketbra{\psi}{\psi}) \geq \frac{4 \gamma - 1}{3}$, which only becomes non-trivial when $\gamma > \frac{1}{4}$.
	For instance, when $\eps_1 = 0.99$, using the current state-of-the-art bounds of $\calF(\ketbra{\psi}{\psi}) \geq \Omega(\gamma^{598})$ from \cite{mehraban2024improvedboundstestinglow} \emph{without the $\approx  10^{-2175}$ multiplicative constant hidden by the big O notation} leads to requiring $\eps_2 < 2.2 \times 10^{-16}$.
	Meanwhile, the previous bounds simply require $\eps_2 < 0.92$.
	When $\eps_1^6 = \frac{1}{4}$ (i.e., when the bounds of \cite{grewal2023improved} stop working and require $\eps_2 = 0$), the bound on $\eps_2$ instead becomes $\eps_2 \leq 2^{-1196} \approx 9.3 \times 10^{-361}$.
	
	As such, for the sake of practicality we match the bounds of \cite{gross2021schur,grewal2023improved} in the close regime when given mixed state inputs.
	
	The following proof works similarly to that of \cref{claim:many-large-p-hat}.
	
	\begin{lemma}[Generalization of {\cite[Lemma 7.6]{grewal2023improved}}]\label{lem:mass-in-M-lower-bound}
		For a $f : \F_2^{2n} \rightarrow \mathbb{R}^{\geq 0}$ such that $0 \leq f(x) \leq \frac{1}{2^n}$ and $\sum_{x \in \F_2^{2n}} f(x) \leq 1$, let $M = \{x \in \F_2^{2n} : 2^n f(x) > \frac{1}{2}\}$.
		Then
		\[
		\sum_{x \in M} f(x) \geq \frac{4 \left(4^n \sum_{x \in \F_2^{2n}} f(x)^3\right) - 1}{3} .
		\]
	\end{lemma}
	\begin{proof}
		Let $\gamma \coloneqq 4^n \sum_{x \in \F_2^{2n}} f(x)^3$ for conciseness.
		We note that $x \not\in M$ implies that $4^n f(x)^2 \leq \frac{1}{4}$, while $x \in M$ implies $4^n f(x)^2 \leq 1$.
		We also observe that \[\sum_{x \not\in M} f(x) \leq 1 - \sum_{x \in M} f(x).\]
		If we assume, for the sake of contradiction, that $\sum_{x \in M} f(x) < \frac{4 \gamma - 1}{3}$.
		Then
		\[
		\gamma = 4^n \left(\sum_{x \in M} + \sum_{x \not \in M}\right) f(x)^3 \leq \sum_{x \in M} f(x) + \frac{1}{4} \sum_{x \not \in M} \leq \frac{3}{4}\sum_{x \in M} f(x) + \frac{1}{4} < \gamma,
		\]
		which is a clear contradiction.
	\end{proof}
	
	\begin{fact}[{\cite[Lemma 4.8]{chen2024stabilizerbootstrappingrecipeefficient}}]\label{fact:m-commute}
		Let $M = \{x \in \F_2^{2n} : 2^n p_\rho(x) > \frac{1}{2}\}$. Then for all $x, y \in M$, $[x, y] = 0$.
	\end{fact}
	
	\begin{theorem}
		\[\text{If } 4^n \sum_{x \in \F_2^{2n}} p_\rho(x)^3 \geq \gamma, \text{ then } \calF(\rho) \geq \frac{4 \gamma - 1}{3}.\]
	\end{theorem}
	\begin{proof}
		Using \cref{lem:mass-in-M-lower-bound}, \[\sum_{x \in M} p_\rho(x) \geq \frac{4 \gamma - 1}{3}\]
		where $M \subset \F_2^{2n}$ generates an isotropic subspace by \cref{fact:m-commute} and the bilinearity of the symplectic product.
		Extending $M$ to a Lagrangian subspace and using \cref{lem:fidelity-lower} completes the proof, noting that $p_\rho(x) \geq 0$.
	\end{proof}

\end{document}